\documentclass[a4paper,UKenglish]{lipics}
 
\usepackage{microtype}

\pdfoutput=1

\usepackage{comment}
\usepackage{graphicx}
\usepackage{complexity}
\usepackage{amssymb}
\usepackage{savesym}
\usepackage{amsmath}
\usepackage{amsthm}
\usepackage{amssymb}
\usepackage{color}
\savesymbol{vec}
\usepackage{wrapfig}
\restoresymbol{eul}{vec}
\newenvironment{lemma-l}[1]{\noindent {\bf Lemma~#1.~}\em }{\smallskip}
\newenvironment{theorem-t}[1]{\noindent {\bf Theorem~#1.~}\em }{\smallskip}
\newtheorem{obs}[theorem]{Observation}

\usepackage{perpage}
\MakePerPage{footnote}

\usepackage[T1]{fontenc}
\usepackage{pxfonts}

\usepackage{subfig}
\newcommand{\lk}[1]{{\color{blue}[LK: #1]}}
\newcommand{\nari}[1]{{\bf [NARI:} #1]}

\DeclareMathAlphabet{\mathcal}{OMS}{cmsy}{m}{n}
\newcommand{\Ff}{\ensuremath{\mathcal F}}
\newcommand{\Cc}{\ensuremath{\mathcal C}}
\newcommand{\Rr}{\ensuremath{\mathcal C}}

\theoremstyle{empty}
\newtheorem{duplicate}{Theorem}
\hypersetup{
    colorlinks=true,
    linkcolor=black,
    citecolor=black,
    filecolor=black,
    urlcolor=black,
}
\begin{document}

\title{Hitting Set for hypergraphs of low VC-dimension}

\author[1]{Karl Bringmann}
\author[2]{L\'{a}szl\'{o} Kozma}
\author[3]{Shay Moran}
\author[4]{N.\ S.\ Narayanaswamy}

\affil[1]{Institut f\"ur Theoretische Informatik, ETH Z\"{u}rich\\
  \texttt{\footnotesize karl.bringmann@inf.ethz.ch}}
\affil[2]{Department of Computer Science, Saarland University\\
  \texttt{\footnotesize kozma@cs.uni-saarland.de}}
\affil[3]{Department of Computer Science, Technion-IIT, Israel, Microsoft Research, Hertzelia, and Max Planck Institute for Informatics, Saarbr\"{u}cken, Germany\\
  \texttt{\footnotesize shaymoran1@gmail.com}}  
\affil[4]{Department of Computer Science and Engineering, Indian Institute of Technology Madras,~~   \texttt{\footnotesize swamy@cse.iitm.ac.in}}  
\authorrunning{K.\ Bringmann, L.\ Kozma, S.\ Moran, N.\ S.\ Narayanaswamy} 

\Copyright{Karl Bringmann, L\'{a}szl\'{o} Kozma, Shay Moran, N.\ S.\ Narayanaswamy}


\serieslogo{}

\maketitle
\begin{abstract}
We study the complexity of the Hitting Set problem in set systems (hypergraphs) that avoid certain sub-structures. In particular, we characterize the classical and parameterized complexity of the problem when the Vapnik-Chervonenkis dimension (VC-dimension) of the input is small. 

VC-dimension is a natural measure of complexity of set systems. Several tractable instances of Hitting Set with a geometric or graph-theoretical flavor are known to have low VC-dimension. In set systems of bounded VC-dimension, Hitting Set is known to admit efficient and almost optimal approximation algorithms (Br\"{o}nnimann and Goodrich, 1995; Even, Rawitz, and Shahar, 2005; Agarwal and Pan, 2014). 

In contrast to these approximation-results, a low VC-dimension does not necessarily imply tractability in the parameterized sense. In fact, we show that Hitting Set is $W[1]$-hard already on inputs with VC-dimension $2$, even if the VC-dimension of the dual set system is also $2$. Thus, Hitting Set is very unlikely to be fixed-parameter tractable even in this arguably simple case. This answers an open question raised by King in 2010. For set systems whose (primal or dual) VC-dimension is $1$, we show that Hitting Set is solvable in polynomial time.

To bridge the gap in complexity between the classes of inputs with VC-dimension $1$ and $2$, we use a measure that is more fine-grained than VC-dimension. In terms of this measure, we identify a sharp threshold where the complexity of Hitting Set transitions from polynomial-time-solvable to NP-hard. The tractable class that lies just under the threshold is a generalization of Edge Cover, and thus extends the domain of polynomial-time tractability of Hitting Set.

\end{abstract}


\section{Introduction}
Let $\Cc$ be a collection of subsets of a finite set $X$. We call the pair $(X, \Cc)$ a \emph{set system}.\footnote{Alternative names in the literature are \emph{hypergraph} and \emph{range space}.} A \emph{hitting set} of $(X,\Cc)$ is a subset of $X$ that has non-empty intersection with all members of $\Cc$. The decision version of the Hitting Set problem asks, given a positive integer $k$, whether a set system has a hitting set of size at most $k$.
 
Hitting Set and its dual, Set Cover, are both ubiquitous and notoriously difficult problems. 
For an arbitrary set system $(X, \Cc)$, Hitting Set is NP-hard to approximate~\cite{lund, ams} with a multiplicative factor better than $c \cdot \log(|\Cc|\cdot |X|)$, for some constant $c>0$. 

Given a set system $\Ff = (X, \Cc)$, and a set $A \subseteq X$, we define the \emph{projection\footnote{Also known in the literature as the \emph{trace} of a set system.} of $\Ff$ on $A$} as $PR_{\Ff}(A) = \{R \cap A \mid R \in \Cc \}$. A set $A$ is said to be \emph{shattered} by $\Ff$ if $PR_{\Ff}(A) = 2^A$, i.e.\ the set of all subsets of $A$. The \emph{Vapnik-Chervonenkis dimension} (or \emph{VC-dimension}) of a set system $\Ff$, denoted $VC(\Ff)$, is the cardinality of the largest set shattered by $\Ff$.
VC-dimension was originally introduced in learning theory~\cite{vc71,blumer}, where it captures the sample complexity in the PAC model.
Since its introduction, VC-dimension has seen many further applications both inside and outside learning theory (see e.g.~\cite{chazelle, matousek}) and it has become a standard measure of complexity of set systems.

Allowing a set system to have large VC-dimension means that less restrictions are placed on its structure, making it more difficult as a Hitting Set instance. In this paper we study both the classical and parameterized complexity of Hitting Set when the VC-dimension of the input set system is bounded.\\

\vspace{-0.1in}
\noindent{\bfseries Hitting Set and parameterized complexity.}
In parameterized complexity, a problem is called {\em fixed-parameter tractable} (FPT) with respect to a parameter\footnote{In this paper we always use the standard parameter, i.e.\ the solution size $k$.} $k$, if there exists an algorithm that solves it in time $O(f(k) \cdot n^{O(1)})$ for an arbitrary function $f$ (where $n$ is the input size). Fixed-parameter tractability has emerged as a powerful tool to deal with hard combinatorial problems. We refer the reader to \cite{DF99,Nie06, FG06} for more details. Unfortunately, Hitting Set is $W[2]$-hard~\cite{FG06}, and thus unlikely to be FPT, meaning that it is hopelessly difficult even from a parameterized perspective. 

However, instances arising in various applications (e.g.\ in graph-theoretical or in geometric settings) often have special structure that can be algorithmically exploited. Indeed, the literature abounds with studies of problems - many of them FPT - that can be seen as special cases of Hitting Set.

Graph-theoretical examples of Hitting Set problems include Vertex Cover, Edge Cover, Feedback Vertex Set, and Dominating Set. In each of these problems the input set system is implicitly defined by an underlying graph $G$, with sets corresponding to the edges, vertices, cycles, and neighborhoods of $G$, respectively. The first three of these problems are well-known to be FPT (Edge Cover is even in P). Dominating Set remains $W[2]$-hard, but is FPT in certain families of graphs (see Table~\ref{tab:examples}). 
Intuitively, Dominating Set is hard because it places very few restrictions on the input: Every set system whose incidence matrix is symmetric can be a Dominating Set instance. Special cases where Dominating Set is FPT include \emph{biclique-free} graphs~\cite{biclique1, biclique2} (a family that contains \emph{bounded genus}, \emph{planar}, \emph{bounded treewidth}, and many other natural classes), \emph{claw-free} graphs~\cite{mnich}, and graphs \emph{with girth at least five}~\cite{RamSau08}. The structure that makes these special cases of Dominating Set tractable can be described in terms of \emph{forbidden patterns} in the adjacency matrix of $G$. For instance, biclique-freeness simply translates to the avoidance of an all-$1$s submatrix of a certain size.
Our work continues this line of investigation: A VC-dimension smaller than $d$ can be interpreted as the avoidance of every matrix with $d$ columns that contain all $2^d$ different boolean vectors in its rows. 

In geometric examples of Hitting Set, the input set system is defined by the incidences between (typically) low complexity geometric shapes, such as points, intervals, lines, disks, rectangles, hyperplanes, etc.
VC-dimension is a natural and useful complexity measure for geometrically defined set systems~\cite{vc71, blumer, learning, epsnet}.
In Table~\ref{tab:examples} we list some representative examples of Hitting Set problems from the literature. \\

\vspace{-0.1in}
\noindent{\bfseries Hitting Set and VC-dimension.} 
Given the difficulty of Hitting Set and the wealth of special cases that are FPT or polynomial-time solvable, it is natural to ask for a general structural property of set systems that guarantees tractability. Such a question has been successfully answered in the field of approximation algorithms: After a series of approximation-results for concrete geometric problems, the landmark result of Br\"{o}nnimann and Goodrich~\cite{bronn} gave an almost optimal\footnote{As a further witness to the difficulty of Hitting Set, \emph{almost optimal} here means a \emph{logarithmic factor of the optimum}, i.e.\ $O(\log{k})$. For more restricted geometric problems better approximation ratios are known, see e.g.\ \cite{clarkson,mustafa_ray}.} approximation algorithm for Hitting Set on set systems with bounded VC-dimension. The algorithm has been further improved by Even at al.~\cite{Even} and recently by Agarwal and Pan~\cite{AgarwalPan}. In this paper we consider this question from a parameterized viewpoint.

In general, the relevance of VC-dimension to Hitting Set has long been known: Low VC-dimension implies the existence of an $\epsilon$-net of small size~\cite{epsnet}. An $\epsilon$-net can be seen as a relaxed form of hitting set in which we are only interested in hitting all sets whose size is at least an $\epsilon$-fraction of the universe size. For set systems with low VC-dimension the size of the fractional hitting set is close to the size of the integral hitting set - this observation is the basis of the approximation-result of Br\"{o}nnimann and Goodrich~\cite{bronn}.\\

\vspace{-0.1in}
\noindent{\bfseries Dual VC-dimension.} The \emph{incidence matrix} of a set system $\Ff = (X,\Cc)$ is a $0/1$ matrix with columns indexed by elements of $X$, and rows indexed by members of $\Cc$. An entry $(A,x)$ of the incidence matrix (where $A \in \Cc$ and $x \in X$) is $1$ if $x \in A$, and $0$ otherwise.  

Given a set system $\Ff$, it is natural to consider its \emph{dual} set system denoted $\Ff^T$, obtained by interchanging the roles of elements and sets (i.e.\ transposing the incidence matrix of the set system\footnote{The transposed incidence matrix may contain duplicate rows, contradicting the definition of a set system. It is safe to discard such duplicates, as this does not affect the VC-dimension or the Hitting Set solution.}). The Hitting Set problem on the dual set system is known as \emph{Set Cover}. The VC-dimension of the dual set system, denoted $VC(\Ff^T)$ is a further natural parameter of set systems. It is well-known that if $VC(\Ff) = d$, then the inequality $VC(\Ff^T) < 2^{d+1}$ holds. \\

{\renewcommand{\arraystretch}{1.3}%
\begin{table} 
\scriptsize\centering 
\begin{tabular}{|>{\raggedright}p{0.33\paperwidth}|>{\centering}p{0.08\paperwidth}|>{\centering}p{0.1\paperwidth}|>{\centering}p{0.09\paperwidth}|}
\hline 
{\bf Graph problem} & FPT status & VC-dimension
\tabularnewline
\hline 
Edge Cover & P & $2$ 
\tabularnewline
Tree-Like Hitting Set~\cite{Guo} & P & $\infty$
\tabularnewline
Vertex Cover & FPT & $ 2$ 
\tabularnewline
Dominating Set (claw-free) \cite{mnich} & FPT & $ \infty$ 
\tabularnewline
Dominating Set (girth $\geq 5$) \cite{RamSau08}& FPT & $ 2$ 
\tabularnewline
Dominating Set (planar) \cite{domplanar} & FPT & $ 4$ 
\tabularnewline
Dominating Set ($K_{t,t}$-free) \cite{biclique1, biclique2} & FPT & $ t + \lceil \log_2{t} \rceil$ -1 
\tabularnewline
Feedback Vertex Set \cite{feedback1,feedback2} & FPT & $ \infty$ 
\tabularnewline

Dominating Set (unit disk) \cite{unitdisk} & $W[1]$-hard & $ 3$  
\tabularnewline
Dominating Set (induced $K_{4,1}$-free) \cite{mnich} & $W[1]$-hard & $\infty$  
\tabularnewline
Dominating Set ($\Delta$-free) \cite{RamSau08}& $W[2]$-hard & $\infty$  
\tabularnewline
\hline

{\bf Geometric problem} & FPT status & VC-dimension 
\tabularnewline
\hline 
Line intervals & P & $ 2$ 
\tabularnewline

Halfplane arrangement in $\mathbb{R}^2$ ~\cite{sariel}& P & $3 $ 
\tabularnewline

Disjoint Rectangle Stabbing \cite{HKLRS} & FPT & $2$ 
\tabularnewline

Pseudoline arrangement & FPT & $2$ 
\tabularnewline

Hyperplane arrangement in $\mathbb{R}^d$ & FPT & $d+1$ 
\tabularnewline

Halfspace arrangement in $\mathbb{R}^3$ [\textsection\,\ref{hardnessproof}]& $W[1]$-hard & $4$ 
\tabularnewline

Collection of unit disks in $\mathbb{R}^2$ \cite{panos} & $W[1]$-hard & $3$ 
\tabularnewline

Collection of unit squares in $\mathbb{R}^2$ \cite{panos} & $W[1]$-hard & $3$ 
\tabularnewline

Rectangle Stabbing \cite{DFR09} & $W[1]$-hard & $3$ 
\tabularnewline
\hline

\end{tabular}\protect\caption{Special cases of Hitting Set in the FPT literature, and their VC-dimension. For hardness results, the values for VC-dimension should be prefixed with ``at least'', for algorithmic results (P and FPT) with ``at most''. The results in the table are discussed in \textsection\,\ref{sec:tabledisc1} and \textsection\ \ref{sec:tabledisc2} of the Appendix.}

\label{tab:examples}
\end{table}

}

\vspace{-0.1in}
\noindent{\bfseries Our results.} 
We study the classical and parameterized complexity of Hitting Set restricted to set systems with small VC-dimension. 
In light of Table~\ref{tab:examples}, there is no clear separation 
at any value of the VC-dimension: Some FPT classes have unbounded VC dimension, while $W[1]$-hard classes with VC-dimension $3$ are known\footnote{To the best of our knowledge, prior to our paper there were no $W[1]$-hard examples known with VC-dimension \emph{or} dual VC-dimension lower than $3$. In fact, we are not aware of $W[1]$-hard examples with \emph{explicitly stated} VC-dimension lower than $4$, see \textsection\,\ref{sec:tabledisc2}. }. However, an FPT result for Hitting Set restricted to VC-dimension $2$ would generalize many known FPT results for special cases of Hitting Set. Hence, we study the existence of a small threshold value of VC-dimension, below which Hitting Set is tractable and at which it becomes intractable (both in the parameterized and in the classical sense). The program of finding such a dichotomy for the FPT complexity of Hitting Set in terms of the VC-dimension has also been proposed by King~\cite{stack}.

In this paper, we show the threshold of tractability to be at the (surprisingly low) value of $2$, i.e., we prove $W[1]$-hardness of Hitting Set restricted to VC-dimension $2$ (even if also the dual VC-dimension is $2$). The phenomenon of a large gap between the complexity of set systems of VC-dimension $1$ and set systems of VC-dimension $2$ also occurs in other areas such as communication complexity, machine learning, and geometry~\cite{DBLP:journals/eccc/AlonMY14,MoranSWY15}.
Moreover, assuming the Exponential Time Hypothesis (ETH) we obtain an almost matching lower bound for the trivial $n^{O(k)}$ algorithm. We prove this result in \textsection\,\ref{sec:hardness}.
\begin{theorem}\label{thm1}
Hitting Set and Set Cover restricted to set systems $\Ff=(X,\Cc)$ with $VC(\Ff) = VC(\Ff^T) = 2$ are $W[1]$-hard. Moreover, if any of these problems can be solved in time $f(k) \cdot |X|^{o(k / \log k)}$, where $f$ is an arbitrary function and $k$ is the solution size, then ETH fails. 
\end{theorem}

\noindent{\em Note.} 
Theorem~\ref{thm1} could be stated with $|X|$ replaced by $|\Cc|$ or $|\Cc|\cdot|X|$, which are perhaps more natural as a measure of input length. 
However, the Sauer-Perles-Shelah lemma (see e.g.~\cite{Sauer72})
states that if $VC(\Ff) = d$, then
$|\Cc| \leq \sum_{j=0}^{d}{|X| \choose j}$.
Therefore, $|X|$ and $|\Cc|$ are within a polynomial factor of each other, which allows us to use $|X|$.


The hardness result of Theorem~\ref{thm1} can be strengthened to set systems with symmetric incidence matrices, i.e.\ the result also holds for Dominating Set. The construction is more involved in that case, and we omit it in this version of the paper.

On the positive side, given a set system $\Ff$, if $VC(\Ff)=1$ \emph{or} $VC(\Ff^T)=1$, we show that Hitting Set is in P. The proof is simple and self-contained (see \textsection\,\ref{sec:simple}). The $VC(\Ff)=1$ case was known prior to this work~\cite{stack}, but we are not aware of a published proof.


 \begin{wrapfigure}[10]{l}[0.1\textwidth]{0.73\textwidth}
\begin{center}  
\vspace*{-0.1in}
\includegraphics[width=0.65\textwidth]{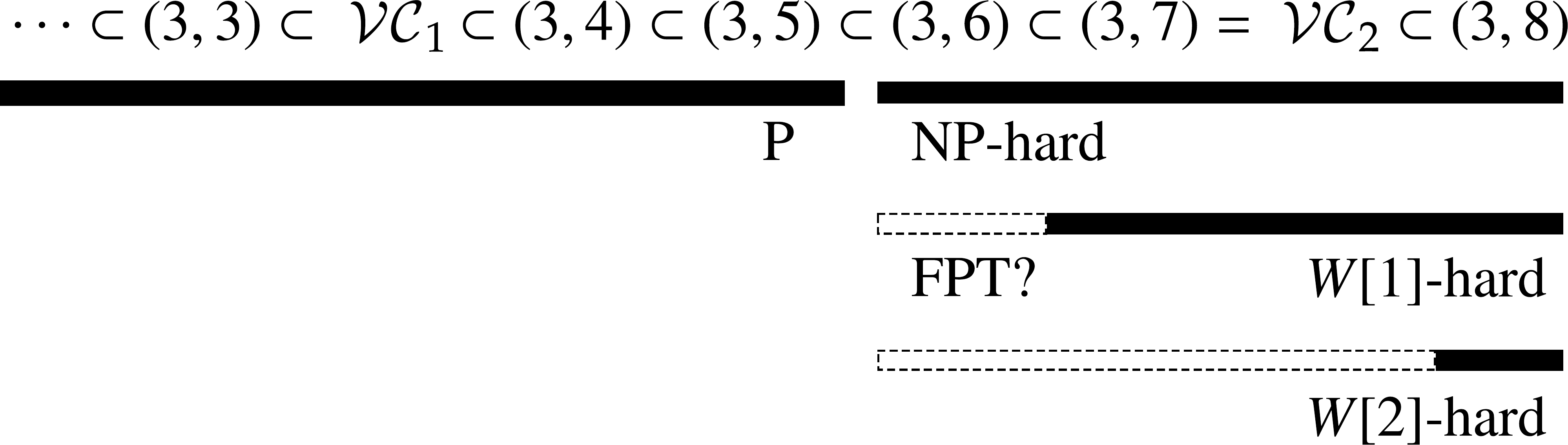}
\end{center}
	\protect\caption{The complexity of Hitting Set when the VC-dimension is low.\label{fig:summary}}
\end{wrapfigure}

To bridge the rather large gap in complexity between set systems of VC-dimension $1$ and $2$, we use a finer parameterization which was also used in~\cite{DBLP:journals/eccc/AlonMY14,MoranSWY15}. For a pair of integers $\alpha,\beta \geq 1$, a set system $\Ff = (X, \Cc)$ is an $(\alpha,\beta)$-system if for any set $A \subseteq X$ with $|A| \le \alpha$ the projection $PR_\Ff(A)$ has cardinality at most $\beta$. In other words, a set system is an $(\alpha, \beta)$-system, if every submatrix of its incidence matrix with $\alpha$ columns has at most $\beta$ different vectors in its rows. Let $\mathcal{VC}_d$ denote the class of set systems with VC-dimension at most $d$. Observe that $\mathcal{VC}_{d-1}$ is equal to the class of $(d,2^d-1)$-systems. Moreover, every set system is a $(d,2^d)$-system, for arbitrary 
$d \geq 1$. Hence, Hitting Set on $(3,8)$-systems is the standard Hitting Set problem (without restrictions) and thus $W[2]$-hard. 
 
The Sauer-Perles-Shelah Lemma can also be stated using this notation: Every $(d,2^d-1)$-system is a $(m,\sum_{j=0}^{d-1}{m \choose j})$-system for every $m\geq d$. In particular, every set system in $\mathcal{VC}_1$ is a $(3,4)$-system. Further, we prove that every Edge Cover instance is a $(3,5)$-system, but the reverse does not hold. Edge Cover is well-known to be solvable in polynomial time using matching techniques~\cite{gareyjohnson}. The next result (see \textsection\,\ref{sec:simple}) extends the domain of polynomial-time solvability from Edge Cover to the larger class of $(3,5)$-systems.
\begin{theorem}\label{thm2}
  Hitting Set on $(3,5)$-systems is in $P$.
\end{theorem}
The algorithm we present for proving Theorem~\ref{thm2} is fairly simple.
However, its analysis is quite involved -- revealing some of the combinatorial structure underlying $(3,5)$-systems.

In contrast to $(3,5)$-systems, it is not hard to see that there are $(3,6)$-systems for which the Hitting Set problem is NP-hard. 
\begin{theorem}\label{thm3}
Hitting Set on $(3,6)$-systems is NP-hard.
\end{theorem}

This discussion yields a complete characterization of the complexity of Hitting Set on $(3,\beta)$-systems with a transition from polynomially-solvable to NP-hard between the $\beta$ values of 5 and 6. Regarding the FPT status of the problem, the picture is almost complete, with the question of $(3,6)$-systems remaining open. The results are illustrated in Figure~\ref{fig:summary}.\\

\vspace{-0.1in}
\noindent{\bfseries Open questions.} An immediate open question raised by our work is whether Hitting Set is FPT on $(3,6)$-systems. The $(\alpha, \beta)$-parameterization provides an ever finer hierarchy of set systems, as $\alpha$ increases. A more ambitious goal would be a full characterization of the complexity of Hitting Set in $(\alpha, \beta)$-systems for $\alpha \geq 3$. 
We only have preliminary results in this direction. Finally, we leave open the question whether Hitting Set is even $W[2]$-hard on set systems of bounded VC-dimension.\\

\vspace{-0.1in}
\noindent{\bfseries Related work.}
Langerman and Morin~\cite{LangermanM05} study the parameterized complexity of an abstract covering problem with a dimension parameter that has some connections to the VC-dimension.
However, the results are not directly comparable with ours: The instances studied by Langerman and Morin can have arbitrarily large VC-dimension and are restricted by other conditions, whereas the instances we study have very low VC-dimension, but have no further constraints.\\

\vspace{-0.1in}
\noindent{\bfseries Notation.}
Consider a set system $\Ff = (X, \Cc)$. Let $b_1,\ldots,b_t \in X$ be distinct elements and $(p_1,\ldots,p_t) \in \{0,1\}^t$. We say that $(b_1,\ldots,b_t)$ \emph{realizes the pattern} $p_1\ldots p_t$ if the set $\{b_i \mid p_i=1 \}$ is contained in $PR_{\Ff}(\{b_1,\ldots,b_t\})$.


\section{Hitting Set with VC-dimension $2$ is $W[1]$-hard}
\label{sec:hardness}
In this section we prove the $W[1]$-hardness of Hitting Set and Set Cover on set systems of VC-dimension $2$ and dual VC-dimension $2$. The NP-hardness of this class was known, implied for example, by the NP-hardness of Vertex Cover.

\begin{duplicate}[restated]
Hitting Set and Set Cover restricted to set systems $\Ff=(X,\Cc)$ with $VC(\Ff) = VC(\Ff^T) = 2$ are W[1]-hard. Moreover, if any of these problems can be solved in time $f(k) \cdot |X|^{o(k / \log k)}$, where $f$ is an arbitrary function and $k$ is the solution size, then ETH fails. 
\end{duplicate}

In the remainder of this section we prove Theorem~\ref{thm1}.
Since Hitting Set on a set system $\Ff$ is equivalent to Set Cover on set system $\Ff^T$, it suffices to prove hardness of Hitting Set.

We reduce to Hitting Set from the Partitioned Subgraph Isomorphism problem: Given a host graph $G = (V,E)$ with a partitioning of the vertices $V = V_1 \cup \ldots \cup V_t$ and a pattern graph $H = ([t],F)$ with $|F| = k$, decide whether there are vertices $u_1 \in V_1,\ldots,u_t \in V_t$ such that $u_i u_j \in E$ for every $ij \in F$. It is known that Partitioned Subgraph Isomorphism is W[1]-hard and cannot be solved in time $f(k) \cdot n^{o(k / \log k)}$, where $n = |V|$, and $f$ is an arbitrary function, unless ETH fails~\cite{Marx07}.

We first show that we may assume the hard instance to have $t=k$, i.e.\ that the number of vertices and the number of edges in $H$ are equal. 
Consider an arbitrary instance of Partitioned Subgraph Isomorphism. Since Partitioned Subgraph Isomorphism splits naturally over connected components of~$H$, we may assume that $H$ is connected. 
If $k > t$, add $k-t$ isolated vertices to $H$, and add the corresponding partitions containing isolated vertices $V_{t+1} = \{v_{t+1}\},\ldots,V_k=\{v_k\}$ to $G$, without changing the existence of a solution. Observe that the parameter $k$ is unchanged. In the case when $k < t$, since $H$ is connected, it follows that $k=t-1$.  We add two components to $H$: a clique on $4$ vertices and an isolated vertex. To $G$ we add the partitions $V_{t+1}=\{v_{t+1}\}, \ldots, V_{t+5}=\{v_{t+5}\}$ such that $v_{t+1},\ldots,v_{t+4}$  form a clique, and $v_{t+5}$ is an isolated vertex.
 After the transformation, $H$ contains $k+6 = t+5$ edges and vertices. 
Furthermore, the equivalence of the solutions is preserved, and the parameter $k$ (the number of edges in $H$), increases by a constant only. 
 

For ease of notation we let $E \subseteq [n] \times [n]$ and write $uv$ for an edge in $E$. Since $G$ is undirected, the set $E$ contains $uv$ if and only if it contains $vu$. Similarly, $F \subseteq [k] \times [k]$ and $ij \in F$ if and only if $ji \in F$. We fix any ordering $<$ on $V$ and the  lexicographic\footnote{{$uv<wz\iff (u<w\lor(u=w\land v<z))$.}} ordering $<$ on $V \times V$ and thus on $E$. We write $E_{ij} := E \cap (V_i \times V_j)$.

We construct an equivalent Hitting Set instance $\Ff$. We start by defining $\Ff$ and proving correctness, and later prove $VC(\Ff) = 2$ and $VC(\Ff^T) = 2$.
\subsection{Construction of $\Ff$}

We construct a set system $\Ff=(X,\Cc)$ as follows. The elements of $X$ are
\begin{align*}
  x_{i,u}^\ell \qquad & \text{ for } i \in [k], \, u \in V_i, \, \ell \in [2 \deg_H(i)],  \\
  y_{ij,uv}^\ell \qquad & \text{ for } ij \in F, \, uv \in E_{ij}, \, \ell \in [5].
\end{align*}
It will be convenient to structure these elements into disjoint \emph{ground sets} $X_i^\ell = \{x_{i,u}^\ell \mid u \in V_i\}$ and $Y_{ij}^\ell = \{y_{ij,uv}^\ell \mid uv \in E_{ij}\}$. We will force each hitting set to pick exactly one element from every ground set; these elements will encode the desired copy of $H$ (should it exist).

In the remainder we define the sets in $\Cc$. First we introduce the following sets of $\Cc$.
\begin{align*}
  A_{i,u}^\ell &= \{ x_{i,v}^\ell \mid v < u \} \cup \{ x_{i,v}^{\ell+1} \mid v \ge u \}, \qquad  \text{ for } i \in [k], u \in V_i,  \, \ell \in [2 \deg_H(i)], \\
  B_{ij,uv}^\ell &= \{ y_{ij,wz}^\ell \mid wz < uv \} \cup \{ y_{ij,wz}^{\ell+1} \mid wz \ge uv \}, \qquad  \text{ for } ij \in F, \, uv \in E_{ij}, \, \ell \in [5]. 
\end{align*}

Here, $x_{i,v}^{\ell+1}$ is to be interpreted as $x_{i,v}^{1}$ for $\ell = 2 \deg_H(i)$, and $y_{ij,wz}^{\ell+1}$ as $y_{ij,wz}^1$ for $\ell = 5$, i.e.\ there is a wrap-around of the index $\ell$.
Note that the disjoint ground sets appear as sets $A_{i,u}^\ell$ (where $u$ is the smallest vertex in $V_i$) and $B_{ij,uv}^\ell$ (where $uv$ is the lexicographic smallest edge in $E_{ij}$). Hence, any hitting set of $\Ff$ contains at least one element of every ground set.

Note that the total number of ground sets is
\begin{align*}
k' =  
5 \, |F| + \sum_{i \in [k]} 2 \deg_H(i) = 9k.
\end{align*}

We set the number of vertices to be chosen in the hitting set to $k'$, i.e.\ from now on we only consider hitting sets of size $k'$ of $\Ff$.
{Since there are exactly $k'$ ground sets, and they are mutually disjoint, it follows that 
any hitting set of $\Ff$ of size $k'$ contains exactly one element $x_{i,u(i,\ell)}^\ell$ of any ground set $X_i^\ell$, and exactly one element $y_{ij,e(ij,\ell)}^\ell$ of any ground set $Y_{ij}^\ell$.}
Moreover, observe that hitting the set $A_{i,u}^\ell$ implies{
$u(i,\ell) < u \vee u(i,\ell+1) \ge u$. 

This holds for all $u\in V$, and so} $u(i,\ell) \le u(i,\ell+1)$ for all $\ell$. Since there is a cyclic wrap-around of~$\ell$ {it follows that $u(i,\ell)=u(i,\ell+1)$ for all $\ell$, and so let $u_i \in V_i$ such that $u(i,\ell) = u_i$ for all $\ell$.} Similarly, the sets $B_{ij,uv}^\ell$ ensure that $e(ij,\ell) = e_{ij}$ for all $\ell$ and some $e_{ij} = v_{ij} w_{ij} \in E_{ij}$.

Observe that the picked edges $e_{ij} = v_{ij} w_{ij}$ form a subgraph of $G$. This subgraph is isomorphic to $H$ if we additionally ensure $u_i = v_{ij}$ and $u_j = w_{ij}$ for all $ij \in F$. To this end, we introduce the sets $C_{ij,u}$ and $C'_{ij,u}$ for $ij \in F$, $u \in V_i$. If $ij$ is the $d$-th edge incident to vertex $i$ in $H$, then we set
\begin{align*}
  C_{ij,u} &= \{ x_{i,v}^{2d-1} \mid v < u \} \cup \{ y_{ij,wz}^1 \mid w \ge u, \, z\in V_j \},  \\
  C'_{ij,u} &= \{ x_{i,v}^{2d} \mid v > u \} \cup \{ y_{ij,wz}^2 \mid w \le u, \, z\in V_j \}.
\end{align*}
Observe that this ensures $u_i = v_{ij}$ for all $ij \in F$. Indeed, fixing $u_i$ the sets $C_{ij,u_i}$ and $C'_{ij,u_i}$ are only hit if we choose $y_{ij,v_{ij} w_{ij}}^1$ with $v_{ij} \ge u_i$ and $y_{ij,v_{ij} w_{ij}}^2$ with $v_{ij} \le u_i$.

We implement the remaining condition $u_j = w_{ij}$ indirectly by introducing the sets 
\begin{align*}
  D_{ij,uv} = \{y_{ij,wz}^3 \mid wz < uv \} \cup \{y_{ij,wz}^5 \mid wz > uv \} \cup \{ y_{ji,vu}^4 \}, \qquad \text{ for } ij \in F, \, i < j, \, uv \in E_{ij}.
\end{align*}
This encodes the formula $e_{ij} \ne uv \vee e_{ji} = vu$ for all $uv \in E_{ij}$, and thus ensures that $v_{ij} = w_{ji}$ and $w_{ij} = v_{ji}$ for all $ij \in F$ with $i<j$ (and thus also for all $ij \in F$ without the condition $i<j$). This indirectly encodes the restriction $u_j = w_{ij}$, since $u_j = v_{ji}$ (by the sets of type $C_{ji,*}$ and $C'_{ji,*}$) and $v_{ji} = w_{ij}$ (by the sets of type $D_{ji,*}$).

In total, any hitting set of $\Ff$ of size $k'$ yields a subgraph of $G$ that is equal to $H$. It is easy to show that the inverse holds as well: If $u_1 \in V_1,\ldots,u_{k} \in V_k$ induce a copy of $H$ in~$G$, then picking the elements $x_{i,u_i}^\ell$ and $y_{ij,u_i u_j}^\ell$ for all $ij,\ell$ yields a hitting set of $\Ff$ of size~$k'$.
This shows the correctness of our construction. 

We show that  $VC(\Ff) = VC(\Ff^T) = 2$ in the next two sections.
Since $k' = O(k)$, $|\Ff| = n^{O(1)}$, and the construction of $\Ff$ can be done in polynomial time, W[1]-hardness of Hitting Set restricted to $VC(\Ff) = VC(\Ff^T) = 2$ follows, and any $f(k') |\Ff|^{o(k'/\log k')}$ algorithm for this problem would yield an $f(k) n^{o(k/\log k)}$ algorithm for Partitioned Subgraph Isomorphism, contradicting ETH.

\subsection{VC-dimension $2$}

It is easy to see that in general $VC(\Ff)$ can be at least 2, e.g., the elements $x_{i,1}^2, x_{i,2}^2$ are shattered by the sets $A_{i,1}^{1}$ (pattern 11), $A_{i,2}^1$ (pattern 01), $A_{i,2}^2$ (pattern 10), and any set of type $B$ (pattern 00).

To prove that $\Ff$ has VC-dimension at most 2, we first argue that we can remove the single element $y_{ji,vu}^4$ from $D_{ij,uv}$, i.e.\ we replace any set $D_{ij,uv}$ by 
$$D_{ij,uv}^* := D_{ij,uv} \setminus \{y_{ji,vu}^4\} = \{y_{ij,wz}^3 \mid wz < uv \} \cup \{y_{ij,wz}^5 \mid wz > uv \},$$ to obtain a set system $\Ff^*$. 
We claim that if there are elements $a,b,c$ realizing the patterns $110,101,011,111$ in $F$ then these elements also realize these patterns in $\Ff^*$. Indeed,
assume for the sake of contradiction that there are elements $a,b,c$ realizing all of the patterns $110$, $101$, $011$, and $111$ in $\Ff$ but not in $\Ff^*$. Then without loss of generality, {for some $ij\in F, uv\in E_{ij}$}, $a = y_{ji,vu}^4$ and $b \in D_{ij,uv}\setminus\{a\} = D^*_{ij,uv}$. Now, there is only one set in $\Ff$ containing both $a$ and $b$, namely $D_{ij,uv}$ {(since $D_{ij,uv}$ is the only set which intersects both $Y_{ji}^4$ and $Y_{ij}^3\cup Y_{ij}^5$).} Thus, one of the patterns $110$ and $111$ is missing, contradicting the assumption that $a,b,c$ realize all patterns $110$, $101$, $011$, and $111$. Hence, if we show that $\Ff^*$ contains no three elements realizing all patterns $110$, $101$, $011$, and $111$, then no three elements of $\Ff$ are shattered.

To this end, we first lift the ordering of $V$ and the lexicographic ordering of $E$ to orderings on the ground sets, i.e.\ for $u < v$ we set $x_{i,u}^\ell < x_{i,v}^\ell$ and for $uv < wz$ we set $y_{ij,uv}^\ell < y_{ij,wz}^\ell$.
We use the following crucial observation about this ordering and $\Ff^*$.

\begin{obs} \label{obs:vctwo}
  Any set system in $\Ff^*$ intersects at most two ground sets.
  Any set system in $\Ff^*$ restricted to any ground set $S$ forms an \emph{interval} (with respect to the ordering on $S$). Moreover, for any pair of ground sets $S_1 \ne S_2$, the sets of $\Ff^*$ intersecting both $S_1$ and $S_2$ either all intersect in the smallest element of $S_1$ or all intersect in the largest element of $S_1$.
\end{obs}

With this observation at hand, consider any elements $a,b,c \in X$. We do a case distinction over the number of different ground sets that $a,b,c$ are contained in.

(1) If $a,b,c$ come from the same ground set $S$, then they are ordered in $S$, say $a<b<c$. Since each set of $\Ff^*$ forms an interval in $S$, there is no set of $\Ff^*$ containing $a$ and $c$ but not $b$. 

(2) If $a$ and $b$ come from the same ground set $S_1$, say with $a<b$, and $c$ comes from a different ground set $S_2$, then we consider the last part of Observation~\ref{obs:vctwo}. If all sets of $\Ff^*$ containing elements of $S_1$ and $S_2$ contain the smallest element of $S_1$, then since these sets form an interval restricted to $S_1$, there is no set of $\Ff^*$ containing $b$ and $c$ but not $a$. We argue similarly if all sets of $\Ff^*$ containing elements of $S_1$ and $S_2$ contain the largest element of $S_1$. 

(3) If $a,b,c$ all come from different ground sets, then no set in $\Ff^*$ contains all three elements, since any set of $\Ff^*$ intersects at most two ground sets. 

In all cases we showed that one of the patterns $110$, $101$, $011$, and $111$ is missing for any elements $a,b,c \in X$. This finishes the proof of $VC(\Ff) \le 2$.

\subsection{Dual VC-dimension $2$}
It is easy to see that in general the dual VC-dimension of $\Ff$ is at least 2, e.g., the sets $A_{i,1}^1, A_{i,2}^1$ are shattered by the elements $x_{i,2}^2$ (pattern $11$), $x_{i,1}^2$ (pattern $10$), $x_{i,1}^1$ (pattern $01$), and any element of the form $y_{ij,uv}^\ell$ (pattern $00$).

To show that the dual VC-dimension of $\Ff$ is at most 2, we first reduce to the set system~$\Ff^*$ like in the previous section.
Consider any sets $M_1,M_2,M_3 \in \Cc$ and assume for the sake of contradiction that they realize all of the patterns $110$, $101$, $011$, and $111$ in $\Ff$ but the corresponding sets $M_1^*,M_2^*,M_3^*$ do not realize all patterns $110$, $101$, $011$, $111$ in $\Ff^*$. 
Without loss of generality, assume that $M_1$ is of the form $D_{ij,uv}$ and its element $y_{ji,vu}^4$ is also contained in $M_2$, so that $y_{ji,vu}^4$ induces one of the patterns $110$ or $111$. This yields that $M_2$ is of the form $B_{ji,wz}^\ell$ for appropriate $\ell \in [5], wz \in E_{ji}$. However, any such set has only one element in common with $D_{ij,uv}$, namely $y_{ji,vu}^4$. Thus, {one of the patterns $110$, $111$ is missing, which is a contradiction.} 
Hence, if we show that $\Ff^*$ contains no three sets realizing all patterns $110$, $101$, $011$, and $111$, then no three sets of $\Ff$ are shattered.

Consider any sets $M_1^*,M_2^*,M_3^*$ of $\Ff^*$ and assume for the sake of contradiction that they realize all of the patterns 110, 101, 011, and 111. 
Restricted to any ground set $S$ the sets $M_1^*,M_2^*,M_3^*$ form intervals, and thus $S$ cannot induce all four patterns $110$, $101$, $011$, and $111$ on $M_1^*,M_2^*,M_3^*$ (as can be checked easily and follows from the proof of the well-known fact that intervals have dual VC-dimension $2$). 

Hence, without loss of generality there is a ground set $S_1$ with an element inducing the pattern $111$ and another ground set $S_2$ with an element inducing the pattern $110$ on $M_1^*,M_2^*,M_3^*$. Note that $M_1^*$ and $M_2^*$ are contained in $S_1 \cup S_2$, since every set of $\Ff^*$ intersects at most two ground sets.
By Observation~\ref{obs:vctwo}, since $M_1^*$ and $M_2^*$ intersect both $S_1$ and $S_2$, they both contain the smallest or largest element $e_1$ of $S_1$ and the smallest or largest element $e_2$ of $S_2$. In particular, restricted to $S_1$ we have without loss of generality $M_1^* \subseteq M_2^*$.
Now, if $M_3^*$ does not intersect $S_2$, then the pattern 101 is missing, since only elements of $S_1$ can be contained in both $M_1^*$ and $M_3^*$, but any such element is also contained in $M_2^*$. 
Otherwise, $M_3^*$ also contains $e_1$ and $e_2$, so that restricted to $S_1$ we have a linear ordering $M_{\pi(1)}^* \subseteq M_{\pi(2)}^* \subseteq M_{\pi(3)}^*$ and restricted to $S_2$ we have a linear ordering $M_{\sigma(1)}^* \subseteq M_{\sigma(2)}^* \subseteq M_{\sigma(3)}^*$ (for permutations $\pi, \sigma$). However, two linear orderings can only induce two of the patterns $110$, $101$, and $011$. This contradicts $M_1^*,M_2^*,M_3^*$ realizing all patterns $110$, $101$, $011$, and $111$, and finishes the proof of $VC(\Ff^T) \le 2$.


\section{Efficiently solvable classes of Hitting Set}
\label{sec:simple}
In this section, we consider efficiently solvable special cases of Hitting Set.
The following result can be seen a warmup for a similar but more involved argument in \textsection\,\ref{sec:35}.

 \begin{theorem}\label{thm:vc1}
Hitting Set is polynomial-time solvable on set systems of VC-dimension $1$ and on set systems of dual VC-dimension $1$.
\end{theorem}

\begin{proof}
Let ${\cal F} = (X,\Cc)$ be a set system of VC-dimension $1$.  If every set in $\Cc$ has non-empty intersection with some $\{x,y\}\subseteq X$ then $\{x,y\}$ is a hitting set of size $2$, and the minimal hitting set can be found by a brute-force search over all subsets of $X$ of size $1$ or $2$.  

Assume therefore that there is no pair $\{x,y\}\subseteq X$ which hits every set in $\Cc$.
Let $x,y\in X$. We say that $x$ \emph{dominates} $y$ if every set in $\Cc$ which contains $y$ also contains $x$. Note that if $x$ dominates $y$, then removing $y$ from all sets in $\Cc$ does not affect the size of the minimum hitting set.
Let $\{x,y\}$ be a two-element set which is contained in some set $A \in \Cc$. 
We claim that $x$ dominates $y$ or $y$ dominates $x$. Indeed, $(x,y)$ realizes the patterns $00$ (by the first observation that no pair $\{x,y\}$ hits every set in $\Cc$) and $11$ (since $\{x,y\} \subset A$). Since $\{x,y\}$ is not shattered, one of $01$ and $10$ must be missing -- implying that one of $x$ or $y$ dominates the other. We proceed by repeatedly removing dominated elements, until we are left with singleton sets which immediately yields the minimum hitting set.

Now consider the case of dual VC-dimension 1. This condition implies that for every pair of sets $A,B\in \Cc$, at least one of the following holds: $A\subseteq B$, $B\subseteq A$, $A\cap B=\emptyset$, or  $A\cup B = X$.

If there exist $A,B\in \Cc$ such that $A \subseteq B$, then we can consider the modified set system in which $B$ is removed, without affecting the size of a minimal hitting set.  Thus, we may assume that no set in $\Cc$ contains another set of $\Cc$. If the sets in $\Cc$ are all pairwise disjoint, then the minimum hitting set contains an arbitrary element from each set, and can easily be found. Thus, we can assume that there exist two sets $A, B \in \Cc$ such that $A \cup B = X$. Any other set $C \in \Cc$ intersects both $A$ and $B$ (otherwise it would be contained in one of them). 
From this we conclude that every $C\in \Cc\setminus\{A,B\}$ satisfies $C\cup A = X$,
and $C\cup B = X$, or equivalently $C$ must contain $B \setminus A$ and $A \setminus B$. It follows that the size of the minimum hitting set is at most $2$, and thus can be computed in polynomial time. \qedhere

\end{proof}

The Sauer-Perles-Shelah Lemma implies that set systems of VC-dimension $1$ are $(k,k+1)$-systems for every $k$, and in particular they are $(3,4)$-systems. Thus, a natural question is whether Hitting Set is polynomial-time solvable for every $(3,4)$-system. We next show that the answer is yes, even for the more general case of $(3,5)$-systems, thus extending Theorem~\ref{thm:vc1}.

\subsection{$(3,5)$-systems}
\label{sec:35}
In this subsection we prove that Hitting Set on $(3,5)$-systems is solvable in polynomial time. Before presenting the algorithm, we briefly observe that the class of $(3,5)$-systems is a proper generalization of Edge Cover instances (i.e.\ where every element occurs in exactly two sets). More generally, an Edge Cover instance is a $(k, k + \lfloor k/2 \rfloor+1)$-system for any $k \geq 1$. This is because the incidence matrix of an Edge Cover instance can have at most $2k$ one-entries in any $k$ columns, and every collection of $k + \lfloor k/2 \rfloor + 2$ distinct $k$-vectors has at least $2k+1$ one-entries. To see that Edge Cover instances are a proper subset of $(3,5)$-systems, observe that in a $(k, k + \lfloor k/2 \rfloor + 1)$-system, an element can occur in an arbitrary number of sets.

\begin{duplicate}[restated]
Hitting set on $(3,5)$-systems is in P.
\end{duplicate}

Let ${\cal F} = (X, \Rr)$ be a $(3,5)$-system. We present a polynomial-time algorithm which outputs a minimum hitting set for $\Ff$.
First check whether $\emptyset\in\Rr$; if this is the case, then  report ``no solution''. Otherwise perform the following preprocessing steps repeatedly, until none of the steps can be performed.
\begin{enumerate}
\setcounter{enumi}{-1}
\item If $\Ff$ is not connected, i.e.\ there are set systems $(X_1,\Cc_1)$,
$(X_2,\Cc_2)$ with disjoint $X_1,X_2$ and $\Ff=(X_1 \cup X_2, \Cc_1 \cup \Cc_2)$, then
recursively solve $(X_1,\Cc_1)$ and $(X_2,\Cc_2)$ and return the union of the
solutions.
\item If $\{x,y,z\} \subseteq X$, and the pattern $000$ is not realized on $(x,y,z)$, then a minimum hitting set is of size at most $3$, and we find it by exhaustive search over all subsets of size at most 3. 
\item If $\{x,y\} \subseteq X$, and the pattern $01$ is not realized on $(x,y)$, then remove $y$ from $X$, as $x$ dominates $y$ (whenever $y$ occurs, $x$ also occurs).  
\item If $A, B \in \Rr$ such that $A \subseteq B$, then remove $B$ from $\Rr$, as whenever we hit $A$, we also hit $B$.
\item If there is a singleton set $\{x\} \in \Rr$, then add $x$ to the solution, remove $x$ from $X$ and remove every set containing $x$ from $\Rr$.
\item (only if steps $0,\ldots,4$ cannot be applied)
 If $A,B,C \in \Rr$, and there is an element $x \in (A \cap B \cap C)$, then add $x$ to the solution, remove $x$ from $X$ and remove every set containing $x$ from $\Rr$.
\end{enumerate}

Observe that after every preprocessing step the resulting set system is still a $(3,5)$-system. Moreover, after the preprocessing, every element of $X$ is contained in exactly two sets of $\Rr$ (otherwise rule 2,4, or 5 is applicable). In other words, after preprocessing, $(X,\Rr)$ is an instance of Edge Cover - such an instance can be solved in polynomial time by computing a maximum matching, and then augmenting with additional edges to cover the unmatched vertices~\cite{gareyjohnson}. The total asymptotic running time (including the time of the preprocessing) is dominated by the time needed to find a maximum matching in a graph with $|\Rr|$ vertices and $|X|$ edges.

The correctness of the algorithm hinges on the validity of the preprocessing steps. Note that only step 5 is not trivially valid. Theorem~\ref{thm2} thus follows from the following claim.
\begin{lemma}\label{claim1}
If preprocessing steps $0,\ldots,4$ cannot be applied, and if there exists an element $x$ contained in at least three sets of $\Rr$, then $x$ is part of any minimum hitting set.
\end{lemma}

\begin{proof}

We make use of the following claim that we prove later.
\begin{lemma}\label{claim2}
If preprocessing steps $0,\ldots,4$ cannot be applied, then for any two sets $A,B \in \Rr$, we have $|A \cap B| \leq 1$.
\end{lemma}
Suppose that there is an element $x \in X$ contained in $t$ sets of $\Rr$, where $t \geq 3$, and let $A_1, \ldots, A_t$ denote the sets containing $x$. Each of these sets must also contain some element other than $x$ (by preprocessing step 4), so let $a_1 \in A_1 \setminus \{x\}, \ldots, a_t \in A_t \setminus \{x\}$. Observe that since $t\geq 3$, Lemma~\ref{claim2} implies that for all $i\neq j$: $A_i\cap A_j=\{x\}$ and therefore, $a_1, \ldots, a_t$ are distinct.

The proof proceeds by showing that every hitting set that does not contain $x$ must contain $a_1,\ldots,a_t$, and that replacing $a_2,\ldots, a_t$ by $x$ preserves the property of being a hitting set.

For every $a_i$ there exists a set $A'_i \in \Rr$ such that $a_i \in A'_i$ and $x \not\in A'_i$, as otherwise $a_i$ would have been deleted in step 2 of the preprocessing, as it is dominated by $x$. 

We show that $A'_i=A'_j$ for all $i,j\leq t$.
Suppose first, towards contradiction, that there exist two indices $i$ and $j$, such that $a_i \notin A'_j$. Let $k$ be an index ($1 \leq k \leq t$) different from $i$ and $j$. In this case, the triple $(x,a_i,a_j)$ realizes the patterns $000$ (by preprocessing step 1), {$100$} (from $A_k$), {$110$} (from $A_i$), {$101$} (from $A_j$), $001$ (from $A'_j$), and either $011$ or $010$ (from $A'_i$), in both cases contradicting the hypothesis that ${\cal F}$ is a $(3,5)$-system. We conclude that for all indices $i$ and $j$, $a_i \in A'_j$.
Thus, we have $\{a_i, a_j\} \subseteq A'_i \cap A'_j$ for all $i,j$. From Lemma~\ref{claim2} we conclude that $A'_i = A'_j$. Let us denote $W = A'_1=\cdots=A'_t$.

Since the above reasoning holds for any element in $A_i \setminus \{x\}$, we even have $W \supset A_i \setminus \{x\}$ for all $i \leq t$. Observe that $|A_i| = 2$, for all $i \leq t$, as otherwise $W$ would intersect $A_i$ in more than one element, contradicting Lemma~\ref{claim2}. See Figure~\ref{fig:cart} for an illustration.

Suppose that there exists a set $Q \in \Rr$, such that $a_i \in Q$, and $Q \neq A_i$, and $Q \neq W$, and let $j$, $k$ be two indices different from $i$. 
Note that since $\lvert Q\cap W\rvert,\lvert Q\cap A_i\rvert\leq 1$, and $a_i\in Q\cap A_i$, and $a_i\in Q\cap W$, it follows that $a_j\notin Q$ and $x\notin Q$. Thus the triple $(x,a_i,a_j)$ realizes the patterns $000$ (by preprocessing step 1), $100$ (from $A_k$), $110$ (from $A_i$), $101$ (from $A_j$), $010$ (from $Q$), and $011$ (from $W$), contradicting the hypothesis that ${\cal F}$ is a $(3,5)$-system. Therefore, $A_i$ and $W$ are the only sets in $\Rr$ containing $a_i$.

Let $H$ be a hitting set that does not contain $x$. Since each of the sets $A_1, \ldots, A_t$ is of size $2$, in order to hit them we must have $a_1,\ldots,a_t\in H$. However, as $a_i$ (for all $i$) is contained only in $A_i$ and $W$, we can improve the solution by removing $a_2,\ldots,a_t$ and adding $x$. In this way, all the sets containing the removed elements are still hit. This means that $H$ is not a minimum hitting set, and thus preprocessing step $5$ is justified. \end{proof}

\begin{wrapfigure}[9]{l}[0.15\textwidth]{0.25\textwidth}
\begin{center}  
\vspace*{-0.25in}
\hbox{\hspace{0.3in}
\includegraphics[width=0.17\textwidth]{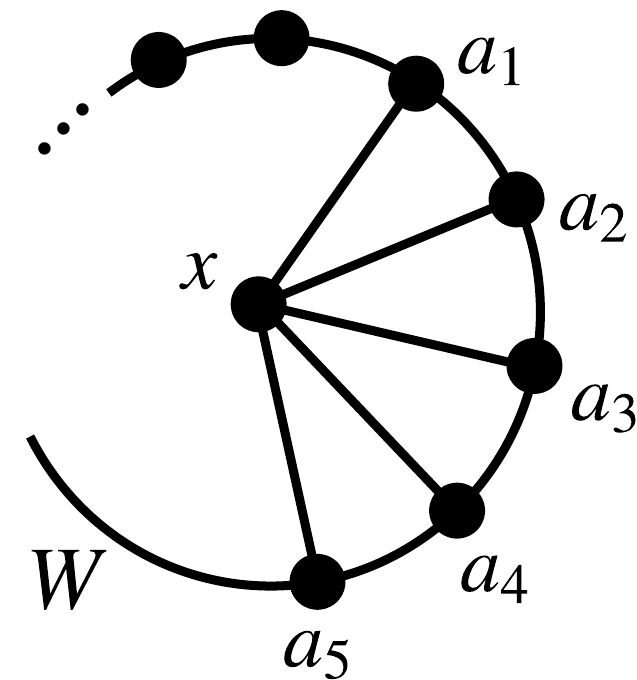}}
\vspace*{-0.25in}
\end{center}
	\protect\caption{Illustration of the proof of Lemma~\ref{claim1}.\label{fig:cart}}
\end{wrapfigure}

It remains to prove Lemma~\ref{claim2}. We proceed via two intermediate claims.


\begin{lemma}  \label{lem:defb4}
 Let $\Ff = (X,\Cc)$ be a $(3,5)$-system such that preprocessing steps $0,\ldots,4$ cannot be applied. Let $Y \subseteq X$ with $|Y| = k\geq 4$. Then the following properties are equivalent. If they are satisfied, we say that $\Ff$ contains a $B_k$-system (induced by $Y$).
  \begin{itemize}
    \item $PR_\Ff(Y)$ contains all subsets of $Y$ of size $k-1$,
    \item $PR_\Ff(Y)$ contains no set $S$ with $0 < |S| \le k-2$.
  \end{itemize}
\end{lemma}
\begin{proof}
  To see that the second property implies the first, observe that all $01$-patterns must be present (by preprocessing step 2), and these patterns can only be realized by having all sets of the form $Y \setminus \{y\}$ for $y \in Y$.
  To see that the first property implies the second, assume for contradiction that there is a set $S \in PR_\Ff(Y)$ with $0<|S| \leq k-2$. Then $S$ realizes pattern $100$ on some $y_1,y_2,y_3$. Since patterns $110$, $101$, $011$, $111$ are realized (by the first property and $k \ge 4$) and $000$ is realized (by preprocessing step 1), we obtain a contradiction.
\end{proof}
\begin{lemma} \label{lem:contb4}
 Let $\Ff = (X,\Cc)$ be a $(3,5)$-system such that preprocessing steps $0,\ldots,4$ cannot be applied. If there are two sets $A,B \in \Cc$ with $|A \cap B| \ge 2$, then there exists a set $Q = \{x,y,z,t\} \subseteq X$ with the following properties:\\ (i)\ \ $Q$ induces a $B_4$-system on $\Ff$, and \\(ii) $Q$ is a hitting set of $\Ff$.
\end{lemma}
\begin{proof}
(i) Consider two elements $x,y \in A \cap B$. By preprocessing step 3 there exist $z \in A \setminus B$ and $t \in B \setminus A$.  
On the triple $(z,x,y)$ we realize $000$ (by preprocessing step 1), $111$ (by $A$), $011$ (by $B$). The missing $01$ patterns on $(x,z)$, $(y,z)$, $(x,y)$, $(y,x)$ can be realized with the assumption that $\Ff$ is a $(3,5)$-system, only if the remaining two patterns on $(z,x,y)$ are $101$ and $110$. A similar argument shows that on the triple $(t,x,y)$ the following five patterns are realized: $000$ (by preprocessing step 1), $111$ (by $B$), $011$ (by $A$), and to obtain $01$ on each pair: $101$ and $110$.
Observe that if $(x,y)$ realizes $00$, $01$, or $10$, then the pattern on $z$ and $t$ is uniquely determined. This yields on the tuple $(z,t,x,y)$ the following patterns: $0000$, $1110$, $1101$ (by joining uniquely the patterns that realize $00$, $01$, and $10$ on $(x,y)$). We need to realize $01$ on both $(z,t)$ and $(t,z)$. The only way to achieve this is with the patterns $0111$ and $1011$ on $(z,x,y,t)$. With this we conclude that $PR_\Ff(\{x,y,z,t\})$ contains all possible sets of size $3$, satisfying the first condition of Lemma~\ref{lem:defb4} and hence $\Ff$ contains a $B_4$ system induced by $\{x,y,z,t\}$.

(ii) Suppose for contradiction that $\{x,y,z,t\}$ is not a hitting set of $\Ff$. Pick $D,D'\in\Rr$ such that $D\cap\{x,y,z,t\}=\emptyset$, $D'\cap D\neq\emptyset$, and $D' \cap \{x,y,z,t\} \neq\emptyset$ (such $D,D'$ exist since $\Ff$ is connected and $\{x,y,z,t\}$ is not a hitting set of $\Ff$).
Let $s \in D \cap D'$.
Observe that $|D' \cap \{x,y,z,t\}| \geq 3$ (by Lemma \ref{lem:defb4}). 
Let $x_1,x_2,x_3$ be distinct elements from $\{x,y,z,t\}$
that belong to $D'$. Let $E\in\Rr$ such that $E$ induces the pattern $110$
on $(x_1,x_2,x_3)$. Such an $E$ exists since $\{x,y,z,t\}$ induce a $B_4$-system on $\Ff$.
We consider two cases: (1) when $s \in E$, and (2) when $s \not\in E$.
(1) On the triple $(s,x_2,x_3)$, we have the patterns $000$ (from the preprocessing), $111$ (by $D'$), $100$ (by $D$), $110$ (by $E$), and to get $01$ on $(x_2,x_3)$ and $01$ $(s,x_2)$, we need at least two more patterns on $(s,x_2,x_3)$, contradicting that ${\cal F}$ is a $(3,5)$-system. 
(2) On the triple $(s,x_1,x_2)$ we have the patterns $000$ (from the preprocessing), $111$ (by $D'$), $100$ (by $D$), and $011$ (from $E$).  To realize $01$ and $10$ on $(x_1,x_2)$, we need two more patterns on $(s, x_1,x_2)$, contradicting that ${\cal F}$ is a $(3,5)$-system. We reached a contradiction, proving that such a $D$ cannot exist, and hence $\{x,y,z,t\}$ is a hitting set of $\Ff$. \qedhere

\end{proof}

Suppose that preprocessing steps $0,\dots,4$ cannot be applied to $\Ff$, and there exist sets $A,B \in \Cc$ with $|A \cap B| \ge 2$. Then, from Lemma~\ref{lem:contb4} it follows that there exists a hitting set $\{x,y,z,t\}$ of $\Ff$, such that $\{x,y,z,t\}$ induces a $B_4$-system in $\Ff$. From the definition of hitting set it follows that $PR_{\Ff}(\{x,y,z,t\})$ does not contain the empty set. From Lemma~\ref{lem:defb4} it follows that $PR_{\Ff}(\{x,y,z,t\})$ does not contain a set of size 1. Hence, on the triple $(x,y,z)$ the pattern $000$ cannot be realized, contradicting the assumption that preprocessing step 1 cannot be applied.

This concludes the proof of Lemma~\ref{claim2} and the proof of correctness for the algorithm.\\

\vspace{-0.1in}

The following theorem gives a sharp threshold on the complexity of Hitting Set by showing the NP-hardness of Hitting Set on $(3,6)$-set systems.
\begin{duplicate}[restated]
Hitting set on $(3,6)$-systems is in NP-hard.
\end{duplicate}
The proof follows by considering the Hitting Set instance that corresponds to Vertex Cover in a triangle-free graph. Indeed, the following two observations establish NP-hardness and the $(3,6)$-property:
(i) Vertex Cover in triangle-free graphs is NP-hard. This can be seen by taking an arbitrary Vertex Cover instance and splitting every edge by adding two internal vertices. The resulting graph is triangle-free. Also, the size of its optimum vertex cover is the original plus the number of edges in the original graph, and (ii) in a triangle-free Vertex Cover instance, on any three elements (vertices) the pattern $111$ and one of the patterns in $\{011, 110, 101\}$ are not realized.

\newpage 
\bibliographystyle{plain}
\bibliography{ref}

\newpage
\appendix

\section{Complexity claims in Table~\ref{tab:examples}}
\label{sec:tabledisc1}
The exact definitions and the complexity results for the various Hitting Set instances can be found by following the respective references. Edge Cover and Vertex Cover are standard problems, described e.g. in~\cite{gareyjohnson}. The study of Hitting Set (a.k.a.\ \emph{transversal} problem) on line intervals goes back to early work of Gallai~\cite{Gallai}. The folklore polynomial-time algorithm follows directly from his combinatorial observations.

The fact that Hitting Set is FPT in set systems defined by pseudolines is folklore, and can easily be explained by the property that any two points are contained in at most one set (i.e.\ line). A generalization of this property holds for arrangements of hyperplanes in $\mathbb{R}^d$: Here any $d$ points are contained in at most one hyperplane. Both properties are subsumed by the biclique-free property, or equivalently the avoidance of a submatrix consisting of all $1$s\footnote{In the literature, the fixed-parameter tractability on biclique-free instances is shown for \emph{Dominating Set}, but the result easily transfers to \emph{Hitting Set}.}.  The definition of the halfspace arrangement problem and a simple proof of hardness in three dimensions is included in \textsection\,\ref{hardnessproof}.

\section{VC-dimension claims in Table~\ref{tab:examples}}
\label{sec:tabledisc2}

The computation of the VC-dimension is an easy exercise for most of the examples in Table~\ref{tab:examples}. We mention that for some of the problems (especially those related to graphs), the value of the VC-dimension seems not to have been explicitly computed in the literature. In some cases this computation leads to approximation-results (via Br\"{o}nnimann and Goodrich~\cite{bronn}) that match the best known approximation ratio obtained via other means. We give a brief overview of the examples listed in Table~\ref{tab:examples}.

The set systems of Vertex Cover and Edge Cover instances are simple: Each set is of size $2$, respectively, each element appears in $2$ sets. In both cases it is easy to see that both the VC-dimension and the dual VC-dimension is at most $2$.

In Tree-Like Hitting Set, the sets are restricted to be subtrees of a tree. Here we can shatter an arbitrary number of elements: Consider the set of all leaves of a tree, and pick any subset of the leaves. Observe that there is a subtree that contains exactly the picked set of leaves and no other leaves. A similar argument holds for the Feedback Vertex Set problem. In a complete graph, color half of the vertices blue, and observe that if we pick any set of blue vertices (possibly the empty set), there can be a cycle containing all the blue vertices of the chosen set and no other blue vertices.

The set system associated with the \emph{Dominating Set} problem is 
the set of all closed vertex neighborhoods of a graph. Since the incidence matrix of this set system is a symmetric square matrix, the dual VC-dimension is the same as the VC-dimension.

\emph{Triangle-free graphs}: The VC-dimension can be arbitrarily large. To see this, consider an independent set $X$ of size $n$, and add $2^n$ further vertices, each connected to a different subset of $X$. Observe that $X$ is shattered, while the graph is triangle-free.

A similar argument holds for \emph{graphs free of induced $K_{t,1}$}, for $t \geq 3$. Consider a $k$-clique $X$ and a $2^k$-clique $Y$, and for each subset $X' \subseteq X$ (including the empty set), connect one vertex of $Y$ to $X'$, and to none of the vertices in $X \setminus X'$. Clearly $X$ is shattered, and thus the VC-dimension is at least $k$. If the constructed graph contains an induced $K_{t,1}$ for $t \geq 3$, then at least two non-connected vertices of the induced subgraph must be both in $X$ or both in $Y$. This is a contradiction, since $X$ and $Y$ are cliques.

\emph{Planar graphs}: A simple case-analysis shows that if a set of five vertices is shattered by the closed vertex neighborhood of a graph, then the graph must contain $K_{3,3}$ or $K_5$ as a subgraph, and thus it cannot be planar. On the other hand, it is easy to construct a planar graph instance where a set of four vertices is shattered.

\emph{Graphs of girth at least $5$}:  A simple case analysis shows that if $3$ vertices are shattered, then the graph has a triangle or a cycle of length $4$. On the other hand, $2$ vertices can be shattered in this graph class. Therefore, the VC-dimension is 2 (see e.g.~\cite{Bousquet}).  An immediate consequence of the boundedness of the VC-dimension is an $O(\log{k})$-factor approximation algorithm for Dominating Set on this class of graphs, as a corollary of the result of Br\"{o}nnimann and Goodrich~\cite{bronn}. A matching result was obtained by Raman and Saurabh~\cite{RamSau08} using sophisticated techniques. 

The claim for \emph{unit disk graphs} follows from simple geometric arguments (see e.g.~\cite{Bousquet}). For \emph{graphs avoiding $K_{t,t}$}, the incidence matrix can not contain a $t$-by-$t$ all-$1$s submatrix. It is easy to check that a matrix with $t + \lceil \log_2{t} \rceil$ columns that contains all possible $0/1$ vectors on its rows contains such a submatrix, whereas a similar matrix with one fewer columns does not. The claim on the VC-dimension follows.

For most geometric set systems in Table~\ref{tab:examples}, the VC-dimension is well known from the computational geometry and learning theory literature. 

\emph{Line intervals}: Given three points on a line, no interval can contain the two outer points without containing the one in the middle. Thus the VC-dimension is at most $2$. If $3$ intervals share a common point, then one interval is in the union of the other two. This ensures that no three intervals can be shattered, thus the dual VC-dimension is at most $2$ as well. Both values are tight.

\emph{Pseudolines}: Since any two sets intersect at most once, a $2$-by-$2$ submatrix of $1$s can not exist in the incidence matrix. This implies that no $3$ points can be shattered by the set system or by its dual. On the other hand, $2$ points can be shattered by both set systems. The claim follows. The boundedness of the VC-dimension yields an $O(\log{k})$-factor approximation algorithm for this problem. A similar result for a special case of the problem was obtained by Grantson and Levcopoulos~\cite{grantson} using different techniques.

\emph{Halfplanes}: It is easy to show that not every subset of size $2$ of a set of $4$ points in the plane can be realized by halfplanes. On the other hand, for $3$ points in general position every subset can be realized. Thus the VC-dimension is $3$. 
Observe that the dual VC-dimension is $2$, since three lines create at most $7$ cells in the plane, therefore not all patterns on $3$ sets can be realized. On the other hand, the $4$ patterns on $2$ sets can be realized.

The claims for \emph{hyperplanes in $\mathbb{R}^d$}, \emph{unit disks} and \emph{unit squares} follow from similar geometric arguments and we omit them.

\emph{Rectangle Stabbing}: In this problem the set system is defined by the incidences between a set of axis parallel rectangles (playing the role of sets), and a set of horizontal and vertical lines (playing the role of elements). Note that at most $4$ lines can be shattered, as three lines with the same orientation can not be shattered (by the same argument as for intervals), thus the VC-dimension is at most $4$ (this can be reached). However, the families of instances constructed in the hardness proof of Dom et al.~\cite{DFR09} have VC-dimension $3$. The same value is obtained for the dual VC-dimension. We omit the details. For \emph{Disjoint Rectangle Stabbing} the VC-dimension is $2$, by an argument similar to the one used for line intervals.

\section{Hardness of Hitting Set for halfspaces in $d \geq 3$}

\label{hardnessproof}

We prove now that the halfspace arrangement problem in $R^d$ is $W[1]$-hard for $d \geq 3$. The halfspace arrangement problem is a special case of Hitting Set, defined as follows:
The input has $n$ points and $n$ halfspaces in $R^d$, and a number $k$. The goal is to select $k$ halfspaces such that each of the given points is contained in at least one of the $k$ halfspaces. The special case for $d=2$ appears in Table~\ref{tab:examples} as \emph{halfplane arrangement}, and is known to be in $P$. 

\begin{theorem}
In $R^d, d \geq 3$, the halfspace arrangement problem is NP-hard and even $W[1]$-hard.
\end{theorem}
\begin{proof}
We reduce from Dominating Set on the intersection graphs of unit disks, which is known to be $W[1]$-hard \cite{unitdisk}. First, observe that the problem stays
$W[1]$-hard if we consider disks with unit radius $r$ on the two-dimensional
sphere $S^2 = \{(x_1,x_2,x_3) \in R^3 : x_1^2+x_2^2+x_3^2 = 1\}$. This is because we can embed any unit disk graph on a tiny
part of the sphere $S^2$ that approximates the plane sufficiently well.
Given the embedding with disk midpoints $p_1,...,p_n$ on $S^2$ and the radius $r$, we want to find a
dominating set of the intersection graph of these disks. This
is equivalent to finding $k$ indices $i_1, \ldots ,i_k$ such that the disks of
radius $2r$ around $p_{i_1},..,p_{i_k}$ cover all points $p_1, \ldots ,p_n$.
We construct an equivalent instance of the halfspace arrangement
problem as follows: The $n$ points are  $p_1, \ldots ,p_n$. Let $0<s<1$. For
$1 \leq i \leq n$ we add a halfspace $H_i = \{x \in R^3 : p_i . x \geq s \}$ with normal vector $p_i$.
Crucially, observe that 
by setting $s$ appropriately, $H_i \cap S^2$ is equal to the disk of
radius $2r$ around $p_i$. 
Since all points $p_i$ lie on $S^2$, it is equivalent whether we
consider $H_i$ or $H_i \cap S^2$ as sets (of the arrangement problem). 
\end{proof}

\end{document}